\documentclass[11pt]{article} 

\usepackage{Liz}
\usepackage{authblk}

\graphicspath{{Figures/}}
\usepackage{cite}

\title{The $\ell^\infty$-Cophenetic Metric for Phylogenetic Trees as an Interleaving Distance}
\author[1]{Elizabeth Munch}
\author[2]{Anastasios Stefanou }
\affil[1]{Dept.~of Computational Mathematics, Science and Engineering; and Dept.~of Mathematics. 
	Michigan State University, East Lansing, MI. \url{muncheli@egr.msu.edu}}
\affil[2]{Dept.~of Mathematics and Statistics. University at Albany -- SUNY, Albany, NY. \url{astefanou@albany.edu}}
\date{}

\begin{document}
\maketitle

\begin{abstract}
There are many metrics available to compare phylogenetic trees since this is a fundamental task in computational biology.  
In this paper, we focus on one such metric, the $\ell^\infty$-cophenetic metric introduced by Cardona et al.
This metric works by representing a phylogenetic tree with $n$ labeled leaves as a point in $\R^{n(n+1)/2}$ known as the cophenetic vector, then comparing the two resulting Euclidean points using the $\ell^\infty$ distance.
Meanwhile, the interleaving distance is a formal categorical construction generalized from the definition of Chazal et al., originally introduced to compare persistence modules arising from the field of topological data analysis.
We show that the $\ell^\infty$-cophenetic metric is an example of an interleaving distance.
To do this, we define phylogenetic trees as a category of merge trees with some additional structure; namely labelings on the leaves plus a requirement that morphisms respect these labels. 
Then we can use the definition of a flow on this category to give an interleaving distance.
Finally, we show that, because of the additional structure given by the categories defined, the map sending a labeled merge tree to the cophenetic vector is, in fact, an isometric embedding, thus proving that the $\ell^\infty$-cophenetic metric is, in fact, an interleaving distance.
\end{abstract}


\section{Introduction}
Phylogenetic trees model the evolutionary relationships among various biological organisms or more general entities 
that evolve through time. 
Comparing two or more phylogenetic trees is a fundamental task in computational biology \cite{Diaconis1998}. 
Studying metrics on phylogenetic trees is of particular importance for phylogenetic tree reconstruction as well as for developing statistics and clustering techniques on phylogenetic trees.
More broadly, comparison techniques of phylogenetic trees find applications in the fields of biology, including bioinformatics, DNA sequences and viral evolution.
There are quite a few metrics for comparison of phylogenetic trees that have been proposed in the literature (e.g.~\cite{Robinson1981,Robinson1979,Bryant2000,Billera2001,Owen2011,Cardona2013,Mailund2004,Alberich2009,Fernau2010,Moulton2015,Valiente2001}, however this is by no means a complete list).
In this paper we focus on the $\ell^\infty$-cophenetic metric on phylogenetic trees which is one of the  $\ell^p$-type of metrics on phylogenetic trees proposed by Cardona et al.~\cite{Cardona2013}.
This metric works by representing a phylogenetic tree as a point in $\R^{n(n+1)/2}$, then giving the distance between two trees as the $\ell^\infty$ distance between the resulting points.

We think of phylogenetic trees as merge trees together with a choice of a labeling on the leaves.
Merge trees are a special case of a more general construction known as the Reeb graph which is one of the basic topics of study in Topological Data Analysis (TDA).  
These structures originally came from the study of Morse functions on manifolds \cite{Reeb1946} and found increased use through the visualization and graphics communities \cite{Biasotti2008}.
However, more recently they have been studied as objects of interest in their own right.
For the purposes of this paper, a Reeb graph is a topological graph $X$ (equivalently, a 1-dimensional stratified space) with a real-valued function $f: X \to \R$ which is monotone on edges.  
Combinatorially, we store this information as a graph with function values defined at the vertices where we interpolate the function linearly on the edges.
A merge tree is a Reeb graph where every vertex has exactly one neighbor with higher function value, and which has one edge whose function values go to $\infty$; we abuse notation and say that this vertex has an endpoint with function value $\infty$. 
See Fig.~\ref{fig:merge-tree} for an example.
\begin{figure}
	\centering 
	\includegraphics[width = .3\textwidth]{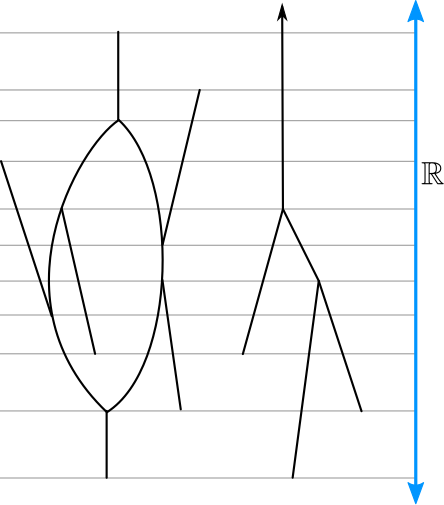}
	\caption[Example of a Reeb graph and merge tree.]{An example of a Reeb graph (left) and a merge tree (middle).  Figures are drawn with the implicit function to $\R$ (at right) given by height.}
	\label{fig:merge-tree}
\end{figure}

The reason for taking this viewpoint is that there has been increased interest in finding metrics for Reeb graphs
\cite{deSilva2016,Bauer2015,DiFabio2016,BauerDiFabioLandi2016,Beketayev2014,Morozov2013,Agarwal2015a,Bauer2014}, 
with a particular view towards understanding properties of a commonly used approximation of the Reeb graph known as mapper 
\cite{singh2007topological,Munch2016,Carriere2017b,Babu2013}. 
Here, we focus on the interleaving distance, which arose from the field of TDA as a method for comparing persistence modules and which generalized the commonly used bottleneck distance for persistence diagrams \cite{Chazal2009b,Chazal2016}.  
Categorified generalizations of these ideas \cite{Bubenik2014,Bubenik2014a,AnastasiosThesis,Silva2017} mean that interleaving distance (strictly, this is an extended pseudometric) can provide new metrics for many different types of input objects.
This extends to Reeb graphs by viewing them as a restricted class of set valued cosheaves over $\R$ \cite{deSilva2016,Curry2014}.
Thus,  merge trees inherit the interleaving distance by virtue of being a subcategory of Reeb graphs, but this be equivalently achieved by viewing merge trees as so-called persistent sets, namely functors $(\R,\leq) \to \Set$ \cite{Morozov2013,AnastasiosThesis}. 
The resulting distance is called the merge tree interleaving distance.

Computing the interleaving distance both on Reeb graphs and merge trees is NP-hard \cite{sidiropoulos9472computing,bjerkevik2017computational}. 
However if we restrict to labeling the vertices on merge trees, e.g.~by considering cluster trees, the complexity of computing the interleaving distance can be significantly improved \cite{Eldridge2015a}.
In this paper we show that computing the interleaving distance on labeled merge trees, i.e.~phylogenetic trees, is polynomial in the number of leaves; see Cor.~\ref{cor:cophenetic-iso}.

In the Sec.~\ref{sec:Bkgd}, we discuss the notion of categories with a flow and equivariant functors, the interleaving distance, and how the $\ell^\infty$-norm can be realized as an interleaving distance on posets.
In Sec.~\ref{sec:CombStructures}, we define merge trees, Reeb graphs and define topological $\e$-smoothings which gives rise to the interleaving distance for merge trees.  
There, we also introduce a combinatorial way to represent merge trees as join-semilattices.
We then define phylogenetic trees as labeled merge trees.
Finally, in Sec.~\ref{sec:CopheneteicIsInterleaving}, we show our main result Thm.~\ref{thm:cophenetic-embedding}, where we realize the $\ell^\infty$-cophenetic metric \cite{Cardona2013} as an interleaving metric on phylogenetic trees using the generalized framework of interleavings on categories with a flow \cite{Silva2017,AnastasiosThesis}. 
This also naturally provides a formula for this interleaving distance which can be computed in polynomial time; see Cor.~\ref{cor:cophenetic-iso}.
Finally, in Sec.~\ref{sec:Discussion}, we discuss future directions for research.

\section{Categorical Structures}
\label{sec:Bkgd}

In this section, we give the necessary category theory-related background.
This background assumes a basic understanding of category theory basics; see, e.g., \cite{MacLane1978} for a good introduction.
This section largely follows the set up and terminology of \cite{Silva2017,AnastasiosThesis}.
\subsection{Categories with a Flow}
\label{ssec:CategoryWithFlow}
Let $\CC$ be a category.
Denote by $\End(\CC)$ the endofunctor category $[\CC,\CC]$.
Also denote by $\R_{\geq0}=(\R_{\geq0},\leq)$ the poset of all nonnegative real numbers.
\begin{dfn}
	\label{dfn:CatWith[0,infty]Action}
	A \textbf{category with a flow}\footnote{This is also known as a $[0,\infty)$-actegory, but category with a flow is both easier to say and fails to generate a flurry of questions about assumed typos.} $(\CC,\Ffunc)$ consists of a category $\CC$, together with
	\begin{itemize}
		\item a functor $\Ffunc:\R_{\geq0}\to\End(\CC)$, $\e\mapsto\Ffunc_\e$, called the \textbf{flow}
		\item a natural transformation $u:I_{\CC}\Rightarrow\Ffunc_{0}$, where $I_{\CC}$ is the identity endofunctor of $\CC$, and
		\item a collection of natural transformations $\mu_{\e,\zeta}:\Ffunc_{\e}\Ffunc_{\zeta}\Rightarrow\Ffunc_{\e+\zeta}$, $\e,\zeta\geq0$,  
	\end{itemize}
	such that the diagrams 
	\begin{equation*}
	\begin{tikzcd}
	\& \Ffunc_{\e}\arrow[ld, swap, Rightarrow, "{u I_{\Ffunc_{\e}}}"]\arrow[rd,  equal, ] \& \&
	\& \Ffunc_{\e}\arrow[ld, swap, Rightarrow, "{I_{\Ffunc_{\e}} u}"]\arrow[rd,  equal, ]\\
	\Ffunc_{0}\Ffunc_{\e}  \arrow[rr, Rightarrow, "{\mu_{0,\e}}"] \& \& \Ffunc_{\e} \&
	\Ffunc_{\e}\Ffunc_{0}  \arrow[rr, Rightarrow, "{\mu_{\e,0}}"] \& \& \Ffunc_{\e}
	\\
	\Ffunc_{\e}\Ffunc_{\zeta}\Ffunc_{\delta}  \arrow[rr, Rightarrow, "{I_{\Ffunc_{\e}}\mu_{\zeta,\delta}}"]\arrow[dd, swap, Rightarrow, "{\mu_{\e,\zeta}I_{\Ffunc_{\delta}}}"]\&\&\Ffunc_{\e}\Ffunc_{\zeta+\delta}\arrow[dd, Rightarrow, "{\mu_{\e,\zeta+\delta}}"] \&
	\Ffunc_{\e}\Ffunc_{\zeta}  \arrow[rr, Rightarrow, "{\mu_{\e,\zeta}}"]\arrow[dd, swap, Rightarrow, "{\Ffunc_{(\e\leq \delta)}\Ffunc_{(\zeta\leq \kappa)}}"' near start]\&\& \Ffunc_{\e+\zeta} \arrow[dd, Rightarrow, "{\Ffunc_{(\e+\zeta\leq \delta+\kappa)}}"' near end]
	\\
	\\
	\Ffunc_{\e+\zeta}\Ffunc_{\delta}\arrow[rr, swap, Rightarrow, "{\mu_{\e+\zeta,\delta}}"] \& \& \Ffunc_{\e+\zeta+\delta}  \&
	\Ffunc_{\delta}\Ffunc_{\kappa}\arrow[rr, swap, Rightarrow, "{\mu_{\delta,\kappa}}"'] \& \& \Ffunc_{\delta+\kappa}
	\end{tikzcd}
	\end{equation*}
	commute for every $\e,\zeta,\delta,\kappa\geq0$. 
	The flow is said to be \textbf{strong} (\textbf{strict}) if the coherence natural transformations $\mu_{\e,\zeta}$ and $u$ are isomorphisms (identities). 
\end{dfn}

We often call the endofunctor $\Ffunc_\e$ the \textbf{$\e$-translation} unless we are in a category where we have a more specific name for it. 
Next, we define maps between categories with a flow.

\begin{dfn}
	\label{definition:equivariant}
	A \textbf{colax equivariant functor} $\Hfunc:\CC\to\DD$ between categories with a flow $\CC=(\CC,\Ffunc,u,\mu)$ and $\DD=(\DD,\Gfunc,v,\lambda)$ is an ordinary functor $\Hfunc:\CC\to\DD$ together with a natural transformation $\eta_{\e}:\mathcal{H}\mathcal{T}_{\e}\Rightarrow\Gfunc_{\e}\mathcal{H}$ for each $\e\geq0$
	such that the diagrams
	\begin{equation*}
	\begin{array}{c}
	\begin{tikzcd}
	\& 
	\Hfunc
	\arrow[ld, swap, Rightarrow, "{I_{\Hfunc} u}"]
	\arrow[rd,  Rightarrow, "{v I_{\Hfunc}}"]
	\\
	\Hfunc\Ffunc_{0}  
	\arrow[rr, Rightarrow, "{\eta_{0}}"] 
	\& \& 
	\Gfunc_{0}\Hfunc
	\end{tikzcd}
	\begin{tikzcd}
	\Hfunc\Ffunc_{\e}  
	\arrow[r, Rightarrow, "{\eta_{\e}}"]
	\arrow[d, swap, Rightarrow, "{I_{\Hfunc}\Ffunc_{(\e\leq \zeta)}}"]
	\&
	\Gfunc_{\e}\Hfunc \arrow[d, Rightarrow, "{\Gfunc_{(\e\leq \zeta)} I_{\Hfunc}}"]
	\\
	\Hfunc\Ffunc_{\zeta} \arrow[r,  Rightarrow, "{\eta_{\zeta}}"] 
	\& 
	\Gfunc_{\zeta}\Hfunc 
	\end{tikzcd}
	\\
	\begin{tikzcd}
	\Hfunc\Ffunc_{\e}\Ffunc_{\zeta}  
	\arrow[r, Rightarrow, "{\eta_{\e} I_{\Ffunc_{\zeta}}}"]
	\arrow[d, swap, Rightarrow, "{I_{\Hfunc}\mu_{\e,\zeta}}"]
	\&
	\Gfunc_{\e}\Hfunc\Ffunc_{\zeta}
	\arrow[r, Rightarrow, "{I_{\Gfunc_{\e}}\eta_{\zeta}}"]
	\&
	\Gfunc_{\e}\Gfunc_{\zeta}\Hfunc 
	\arrow[d, Rightarrow, "{\lambda_{\e,\zeta} I_{\Hfunc}}"]
	\\
	\Hfunc\Ffunc_{\e+\zeta}\arrow[rr,  Rightarrow, "{\eta_{\e+\zeta}}"] 
	\&\&
	\Gfunc_{\e+\zeta}\Hfunc
	\end{tikzcd}
	\end{array}
	\end{equation*}
	commute for all $\e,\zeta\geq0$. 
If all $\eta_{\e}$ are natural isomorphisms (identities) then $\Hfunc$ is called a \textbf{strong} (\textbf{strict}) equivariant functor. 
\end{dfn}

The collection of all categories with a flow together with the colax equivariant functors forms a category on its own which we denote by $\textbf{Flow}$. 
\subsection{The Interleaving Distance associated to a Category with a Flow}
We consider the following generalized setting for a proper notion of a distance.
\begin{dfn}Let $X$ be a set or more generally a proper class. 
	A function $d:X\times X\to [0,\infty]$ is said to be an \textbf{extended pseudometric on $X$} if 
	\begin{itemize}
		\item $d(x_1,x_1)=0$ for all $x_1$ in $\X$ and
		\item $d(x_1,x_3)\leq d(x_1,x_2)+d(x_2,x_3)$ for every $x_1,x_2,x_3$ in $X$.
	\end{itemize} 
	In particular $(X,d)$ is called a \textbf{Lawvere metric space}. 
\end{dfn}

Note that this definition both allows for the possibility that the distance between two objects is $\infty$, and the possibility that $d(x_1,x_2) = 0$ even if $x_1 \neq x_2$.

\begin{dfn}
	Define a morphism $f:(X,d_X)\to(Y,d_y)$ of Lawvere metric spaces is said to be a 1-Lipschitz map if
	\begin{equation*}
	d_Y(f(x_1),f(x_2))\leq d_X(x_1,x_2)\text{ for every }x_1,x_2\in X.
	\end{equation*} 
\end{dfn}

The collection of all Lawvere metric spaces together with 1-Lipschitz maps forms a category which we denote by $\mathbf{Law}$.

Let $(\CC,\Ffunc)$ be a category with a flow. 
The flow $\Ffunc$ on $\CC$ enables us to measure `how far' two objects in $\CC$ are from being isomorphic up to a coherence natural transformation.

\begin{dfn}\label{definition:interleavings}
	Let $X,Y$ be two objects in $\CC$. 
	A \textbf{weak $\e$-interleaving of $X$ and $Y$}, denoted $(\varphi,\psi)$, consists of a pair  of morphisms $\varphi:X\to\Ffunc_{\e}Y$ and $\psi:Y\to\Ffunc_{\e}X$ in $\CC$ such that the following pentagons
	\begin{equation}
	\label{eq:interleaving}
	\begin{tikzcd}
	\Ffunc_{0}X
	\arrow[dd, "{\Ffunc_{(0\leq2\e),X}}"']
	\&
	X
	\arrow[l, "{u_{X}}"']
	\&
	Y
	\arrow[r,"{u_{Y}}"]
	\arrow[dl,  "{\psi}"' very near start]
	\&
	\Ffunc_{0}Y
	\arrow[dd, swap, "{\Ffunc_{(0\leq2\e),Y}}"']
	\\
	\&
	\Ffunc_{\e}X
	\&
	\Ffunc_{\e}Y
	\arrow[leftarrow,  ul, crossing over, "{\varphi}"' very near end]
	\arrow[dl,  "{\Ffunc_{\e}\psi}" very near start]
	\\
	\Ffunc_{2\e}X
	\&
	\Ffunc_{\e}\Ffunc_{\e}X
	\arrow[l, "{\mu_{\e,\e,X}}"' ]
	\&
	\Ffunc_{\e}\Ffunc_{\e}Y
	\arrow[leftarrow,  ul, crossing over, "{\Ffunc_{\e}\phi}" very near end]
	\arrow[r, swap, "{\mu_{\e,\e,Y}}"' ]
	\&
	\Ffunc_{2\e}Y
	\end{tikzcd}
	\end{equation}
	commute. 
	We say that $X,Y$ are \textbf{weakly $\e$-interleaved} if there exists a weak $\e$-interleaving $(\varphi,\psi)$ of $X$ and $Y$. 
	The \textbf{(weak) interleaving distance} with respect to $\Ffunc$ for a pair of objects $X,Y$ in $\CC$ is defined to be
	$$d_{(\CC,\Ffunc)}(X,Y)=\inf\{\e\geq0 \mid X,Y\text{ are weakly }\e\text{-interleaved}\}.$$
	If $X$ and $Y$ are not weakly interleaved for any $\e$, we set $d_{(\CC,\Ffunc)}(X,Y) = \infty$.
\end{dfn}
The ``weak'' moniker is meant to differentiate this definition of interleavings from the traditional persistent homology definitions \cite{Chazal2016}.
There, the fact that the category of persistence modules is a strict 
 category with a flow means that the pentagons of Diagram \ref{eq:interleaving} collapse down to triangles.
In this paper, however, we will drop the word ``weak'' and just reference them as interleaving distances.

As studied and proved in \cite{Silva2017,AnastasiosThesis}, we have the following theorem.

\begin{thm}[{\cite[Thm.~5.3]{Silva2017}}]
	\label{thm:Functoriality}
	Given a category with a flow $(\CC,\Ffunc)$, the associated interleaving distance $d_{(\CC,\Ffunc)}$ forms an extended pseudometric on the objects of $\CC$. 
	Furthermore, the assignment $\mathcal{I}:\textbf{Flow}\to\textbf{Law}$, $(\CC,\Ffunc)\mapsto(\Ob(\CC), d_{(\CC,\Ffunc)})$ is by itself functorial, i.e.
	\begin{align*}
	\textbf{Flow} &\xrightarrow{\hspace{1em}\mathcal{I}\hspace{1em}} \textbf{Law}\\
	\text{Categories with a flow} &\xmapsto{\hspace{2.2em}} \text{Lawvere metric spaces}\\
	\text{Colax equivariant functors} &\xmapsto{\hspace{2.2em}} \text{1-Lipschitz maps}
	\end{align*}
\end{thm}

This theorem is particularly useful due to the following corollary, which says that we need only find a fully faithful colax equivariant functor between categories with a flow in order to obtain an isometric embedding.

\begin{cor}
	\label{cor:Isometric Embedding}
	If $\Hfunc$ is a fully faithful colax equivariant functor $\Hfunc:\CC\to\DD$ between categories with a flow $\CC=(\CC,\Ffunc,u,\mu)$ and $\DD=(\DD,\Gfunc,v,\lambda)$ then it is an isometric embedding with respect to the interleaving distances.
\end{cor}

\begin{proof}
	Consider the image category $\mathbf{Im}\, \Hfunc$, namely the full subcategory of $\DD$ whose objects are images $\Hfunc(X)$ of objects $X$ in $\CC$. 
	Then because $\Hfunc$ is fully faithful the colax equivariant functor $\Hfunc:\CC\to\mathbf{Im}\, \Hfunc$, $a\mapsto \Hfunc(a)$ is an equivalence of categories. 
	So by applying Thm.~\ref{thm:Functoriality} to both $\Hfunc$ and its inverse functor $\Hfunc^{-1}$, we obtain that $\Hfunc:(\Ob(\CC),d_{(\CC,\Ffunc)})\to(\Ob(\mathbf{Im}\, \Hfunc),d_{(\DD,\mathbf{S})})$ is an interleaving isometry.
	In other words $\Hfunc:(\Ob(\CC),d_{(\CC,\Ffunc)})\to(\Ob(\DD),d_{(\DD,\mathbf{S})})$ is an isometric embedding with respect to the interleaving distances.
\end{proof}

\subsection{Interleaving Distances on Posets with a Flow}
\label{ssec:Posets}
In the special case where our category is actually a poset, the interleaving distance becomes much easier to understand.
\begin{dfn}
	A category $\PP$ is said to be a \textbf{poset} if for every $X,Y$ in $\PP$ there exists at most one morphism $f$ from $X$ to $Y$; i.e. the  Hom-set $\Hom_{\PP}(X,Y)$ is either a singleton or the empty set.
\end{dfn}
Let $(\PP,\Omega)$ be a poset with a flow,
and let $d_{(\PP,\Omega)}$ be the interleaving distance on $\PP$ induced by $\Omega$ (Defn.~\ref{definition:interleavings}).
The extra structure of the poset category makes characterizing the interleaving distance rather simple. 
Given two objects $X,Y$ in $\PP$ a pair $(\varphi,\psi)$ of morphisms $\varphi:X\to \Omega_\e Y$ and $\psi:Y\to\Omega_\e X$ is automatically an $\e$-interleaving of $X,Y$ because there exists at most one morphism from $X$ to $\Omega_{2\e} X$ (and at most one from $Y$ to $\Omega_{2\e} Y$  respectively).
So, the interleaving distance on $\PP$ induced by $\Omega$ is given by 
\begin{equation*}
d_{\PP}(X,Y)=\inf\{\e\geq0\mid \exists\; \varphi:X\to\Omega_{\e}Y\text{ and }\psi:Y\to\Omega_{\e}X\}.
\end{equation*}
Posets also satisfy the following interesting property.
\begin{prop} 
	\label{Prop:Poset}
	Let $((\PP,\leq_\PP),\Ffunc)$ and $((\QQ,\leq_\QQ),\Gfunc)$ be posets with a flow. 
	If we have a function $\Hfunc:\Ob(\PP)\to\Ob(\QQ)$, $X\mapsto\Hfunc(X)$, also for every pair of objects $X,Y$ in $\PP$ a function $\Hfunc_{X,Y}:\Hom_\PP(X,Y)\to\Hom_\QQ(\Hfunc(X),\Hfunc(Y))$, $f\mapsto \Hfunc[f]$ and for each $X$ in $\PP$ an inequality $\Hfunc\Ffunc_\e(X)\leq_\QQ \Gfunc_\e\Hfunc(X)$ in $\QQ$, then $\Hfunc$ forms a fully faithful colax equivariant functor. 
\end{prop}
\begin{proof} Because of the poset structure of both categories $\PP$ and $\QQ$, the function $\Hfunc_{X,Y}$ is bijective and the morphism  $\Hfunc\Ffunc_\e(X)\leq_\QQ\Gfunc_\e\Hfunc(X)$ makes $\Hfunc$ a colax equivariant functor.
\end{proof}
We will make use of this proposition in the setting on phylogenetic trees to show the existence of an isometric embedding.


\subsection{The $\ell^\infty$-distance on $\R^n$ is an interleaving distance}
Let $\R^{n}$ be the set of all $n$-tuples of real numbers. 
The $\ell^{\infty}$-norm on $\R^n$ is defined as follows. 
Let $a=(a_1,\ldots,a_n)$ and $b=(b_1,\ldots,b_n)$ be two $n$-tuples in $\R^n$. 
Then define 
\begin{equation*}
\|a-b\|_{\infty}=\max\{|a_i-b_i| : i=1,\ldots,n\}
\end{equation*}
This metric also can be realized as an interleaving distance. 
Consider $\R^n$ as the poset $(\R^{n},\leq)$ where  $a \leq b$ when $a_i\leq b_i$ for all $i=1,\ldots,n$. 

Let $\e\geq0$ and, for ease of notation, let $a + \e = (a_1+\e, \cdots, a_n + \e)$.
Define the $\e$-translation $\Omega_{\e}:(\R^{n},\leq)\to(\R^{n},\leq)$ given by the $\e$-shift upward $a\mapsto a+\e$.
We easily check that $\Omega$ forms a strict flow on $(\R^{n},\leq)$. 
Denote the associated interleaving distance by $d_{(\R^{n},\leq)}$.
Then we have the following result.
\begin{prop}[{\cite[Thm.~3.9]{Silva2017}}]
	The interleaving distance on $\R^{n}$ induced by the strict flow $\Omega$, coincides with the $\ell^{\infty}$-distance on $\R^{n}$.
	That is, for any $a,b \in \R^n$, 
	\begin{equation*}
	d_{(\R^{n},\leq)} (a,b) = \|a-b\|_\infty.
	\end{equation*}
\end{prop}

\begin{obs}
	\label{obs:geq}
	Note that alternatively if $\R^n$ has the poset structure $(\R^n,\geq)$, then we can consider the flow $\Omega=(\Omega_{\e})_{\e\geq0}$ on $\R^n$ given by the $\e$-shift downward $(a_1,\ldots,a_n)\mapsto (a_1-\e,\ldots,a_n-\e)$ to obtain again the $\ell^\infty$-distance.
\end{obs}

\section{Combinatorial Structures}
\label{sec:CombStructures}
In this section, we describe the combinatorial objects of study, in particular merge trees and phylogenetic trees.
\subsection{Merge Trees}
First we define Reeb graphs.
As a first step, we consider  the category $\Rtop$ of \textbf{$\R$-spaces} as defined in \cite{deSilva2016}; 
these are topological spaces $\X$ together with a real valued function $f:\X\to\R$, denoted by $(\X,f)$.
A morphism $\varphi:(\X,f)\to(\Y,g)$ of $\R$-spaces is  a continuous map $\varphi:\X\to\Y$ such that $\varphi\circ g=f$.

\begin{dfn}
	\label{dfn:ReebGraph}
	An $\R$-space $(\mathbb{X},f)$ is said to be a \textbf{Reeb graph} if it is isomorphic to an $\R$-space $(\X,f)$ constructed in the following way.
	Let $S = \{ a_1 < \cdots < a_n\} \subset \R$ be given, called a \textbf{critical set}.  
	\begin{itemize}
		\item  For $i=1,\ldots,n$, we specify a finite set of vertices $\V_i$, which lie over $a_i$
		\item For $i=1,\ldots,n-1$, we specify a finite set of edges $\E_i$ which lie over $[a_i, a_{i+1}]$.
		\item For $i = 0$ and $i = n$, we specify two finite set of edges (possibly empty) $\E_0$ and $\E_n$ lying over $(-\infty, a_1]$ and $[a_n,\infty)$ respectively.
		\item For $i=1,\ldots,n$, we specify left attaching maps $\ell_i : \E_i \to \V_i$ 
		\item For $i = 0, \cdots, n-1$, we specify right attaching maps $r_i:\E_i\to \V_{i+1}$.
	\end{itemize} 
	The space $\X$ is the quotient of the disjoint union 
	\begin{equation*}
	\coprod_{i=1}^n ( \V_i \times \{a_i\} )
	\coprod_{i=1}^{n-1} (\E_i \times [a_i , a_{i+1}]) 
	\coprod (\E_0 \times (-\infty, a_1]) 
	\coprod( \E_n \times [a_n,\infty)) 
	\end{equation*}
	with  respect to the identifications $(\ell_i(e), a_i) \sim (e, a_i)$ and $(r_i(e), a_{i+1}) \sim (e, a_{i+1})$, with the map $f$  being the projection onto the second factor. 
\end{dfn}

We denote by $\Reeb$ the full subcategory of $\Rtop$ whose objects are Reeb graphs.
Reeb graphs are naturally equipped with a strong flow $U$ called \textbf{topological smoothing} \cite{deSilva2016}.
While this flow can be defined combinatorially, it is unnecessary to go into detail on the most general case here as we will immediately restrict our attention to the subcategory of  merge trees.

\begin{dfn}
	A Reeb graph $(\X,f)$ is said to be a \textbf{merge tree} if $\ell_i = \id$ for all $i$ and topmost $E_n$ is the set with one element. 
	Merge trees form a full subcategory of $\Reeb$ denoted by $\Merge$. 
\end{dfn}
\begin{obs}
	\label{obs:merge-smoothing}
	Because merge trees have only \raisebox{\depth}{\rotatebox{180}{Y}}-type interior vertices and an infinite upper tail, the topological $\e$-smoothing $U_\e$ can be thought of as simply shifting each point in the tree $\e$-units downward making a new merge tree. 
Specifically, this is means $U_\e(\X,f) \cong (\X,f_\e)$ 
where $f_\e(x) = f(x) - \e$. 
\end{obs}
Merge trees are closed under the topological smoothings $U_\e$, and hence, merge trees together with the topological smoothings form a category with a strong flow $(\Merge,U)$.
We denote by $d_{\Merge}$ the interleaving distance induced by the strong flow $U$ on $\Merge$.

Though it follows from Thm.~\ref{thm:Functoriality}, it was originally proved by Morozov et al.~that 
\begin{thm}[{\cite[Lem.~1]{Morozov2013}}]
	$d_I^\Merge$ is an extended pseudometric on $\Merge$.
\end{thm}

\subsection{Merge Trees as Posets}

Notice that for a merge tree $(\X,f)$, every pair of points $x,y$ in $\X$ has a unique path $\gamma$ from $x$ to $y$.  
We say this path is \textbf{monotone increasing with respect to $f$} if $t \leq s$ implies $f\gamma(t) \leq f\gamma(s)$.
Define $x \leq_f y$ if the unique path from $x$ to $y$ is monotone increasing.
This gives a poset on the points of $\X$.
In particular, because $(\X,f)$ is a merge tree, for every $x,y$ in $\X$ there exists a unique vertex $v$ of minimum height that connects with $x$ and $y$ known as the \textbf{least common ancestor of $x$ and $y$}. 
We denote this by $x\vee y$. 
The least common ancestor $x\vee y$ of $x$ and $y$ satisfies the following two properties:
\begin{enumerate}
	\item $x\leq x\vee y$ and $y\leq x\vee y$
	\item if $x\leq z$ and $y\leq z$, then $x\vee y\leq z$.
\end{enumerate}
This gives us an operation on $\X$ called the \textbf{join}.
Thus the poset $(\X,\leq_f)$ forms in particular a join-semilattice $(\X,\leq_f,\vee)$.

\begin{obs}
	\label{obs:meet}
	Given a Reeb graph $(\X,f)$ we can define a poset structure in the same way as for merge trees. Also we can define a join $\vee$ as well as a meet operation $\wedge$. 
	From the prospect of $(\X,f)$ being a merge tree however, the meet $\wedge$ is a trivial operation to define because, given two points $x,y$ in the merge tree $\X$, if we assume that their meet $x\wedge y$ exists then either $x\wedge y=x$ or $x\wedge y=y$ (in other words $x\leq_f y$ or $y\leq_f x$).
	Because of that we will not use the meet operation again in this paper. 
\end{obs}

\begin{dfn}
	A vertex $v$ in a merge tree $(\X,f)$ is said to be a \textbf{leaf} if for every $x\in\X$, $x\leq_f v\Rightarrow x=v$.
	Denote by $L(\X,f)$ the set of all leaves of a merge tree $(\X,f)$.
\end{dfn}

\subsection{Phylogenetic Trees with $n$-leaves}
Fix a positive integer $n$. 
\begin{dfn}
	A \textbf{phylogenetic tree with $n$-leaves}, denoted by $(\X,f,\ell)$, is a merge tree $(\X,f)$ together with a bijection $\ell:\{1,\ldots,n\}\to L(\X,f)$, $i\mapsto \ell(i)$ called the \textbf{labeling}.
\end{dfn}

\begin{dfn}
	\label{dfn:MorphPhTrees}
	A \textbf{morphism} $(\X,f,\ell)\xrightarrow{\varphi}(\Y,g,\mu)$ of phylogenetic trees with $n$-leaves is a map $\varphi:\X\to\Y$ such that it is
	\begin{enumerate}
		\item function preserving; $g\circ\varphi=f$;  and 
		\item label preserving: $\mu(i)\leq_g\varphi(\ell(i))$ for all $i=1,\ldots,n$.
	\end{enumerate}
\end{dfn}
With this notion of morphisms, the collection of all phylogenetic trees with $n$ leaves forms a category $\mathbf{PhTree}_{n}$.
We now state some simple but helpful properties of these morphisms.
\begin{prop}
	\label{prop:MorphismProperties}
	If $(\X,f,\ell)\xrightarrow{\varphi}(\Y,g,\mu)$ is a morphism of phylogenetic trees with $n$ leaves, then
	\begin{enumerate}
		\item[i.] if $x_1\leq_f x_2$ then $\varphi(x_1)\leq_g\varphi(x_2)$ for every $x_1,x_2$ in $\X$ 
		\item[ii.]  $\varphi(x_1)\vee_g\varphi(x_2)\leq \varphi(x_1\vee_f x_2)$ for every $x_1,x_2$ in $\X$ 
		\item[iii.]   $\mu(i)\vee_g\mu(j)\leq_g\varphi(\ell(i)\vee_f\ell(j))$ for all $1\leq i\leq j\leq n$.
	\end{enumerate}
\end{prop}
\begin{proof} (i) The first property follows directly by definition of $\leq_f$ and the fact that $\varphi:\X\to\Y$ is function preserving map. 
	(ii) Let $x_1,x_2$ in $\X$.
	Then by (i), the inequalities $x_1,x_2\leq_f x_1\vee x_2$ imply that $\varphi(x_1),\varphi(x_2)\leq_g \varphi(x_1\vee x_2)$.
	Then by the second property of the join, we obtain 
	$$\varphi(x_1)\vee_g\varphi(x_2)\leq_g \varphi(x_1\vee_f x_2).$$ 
	(iii) For the third property, $\mu(i)\leq_f \varphi(\ell(i))$ and $\mu(j)\leq_f \varphi(\ell(j))$ imply  $\mu(i),\mu(j)\leq_g\varphi(\ell(i))\vee_g\varphi(\ell(j))$ and property (ii) gives $\mu(i),\mu(j)\leq_g\varphi(\ell(i)\vee_f\ell(j))$. 
	Finally the second property of the join gives $\mu(i)\vee_g\mu(j)\leq_g\varphi(\ell(i)\vee_f\varphi\ell(j))$ as desired.
\end{proof}

\begin{prop}
	The category $\mathbf{PhTree}_{n}$ is a poset. 
\end{prop}
\begin{proof}
	Assume that $\varphi,\psi\in\Hom_{\mathbf{PhTree}_n}((\X,f,\ell),(\Y,g,\mu))$. 
	We will show that $\phi=\psi$.
	
	Let $x\in\X$. 
	Then there exists a leaf $\ell(i)$ for some index $i=1,\ldots,n$ such that $\ell(i)\leq_f x$. 
	By applying $\varphi$ and $\psi$ we obtain $\mu(i)\leq_g\varphi(\ell(i))\leq_g\varphi(x)$ and 
	$\mu(i)\leq_g\psi(\ell(i))\leq_g\psi(x)$. 
	Thus $\mu(i)\leq_g \varphi(x)\wedge\psi(x)$
	By Obs.~\ref{obs:meet}, $\varphi(x)\leq_g\psi(x)$ or $\psi(x)\leq_g\varphi(x)$.
	Since $\varphi,\psi$ are both function preserving then we have that $\varphi(x)=\psi(x)$.
	Therefore $\varphi=\psi$.
\end{proof}
Because $\mathbf{PhTree}_{n}$ is a poset we will denote the morphisms $(\X,f,\ell)\xrightarrow{\varphi}(\Y,g,\mu)$ in $\mathbf{PhTree}_{n}$ by $(\X,f,\ell)\preceq(\Y,g,\mu)$ for simplicity.
\begin{prop}
	The restriction of the strong flow $U$ on phylogenetic trees forms a strong flow on $\mathbf{PhTree}_n$.
\end{prop}
\begin{proof} 
	We only need to show that the topological $\e$-smoothing $U_\e$ on a phylogenetic tree $(\X,f,\ell)$ is label preserving. 
	Indeed, after the $\e$-smoothing each point $x$ in $\X$ is simply shifted $\e$-units downwards to a point $y$ which satisfies $y\leq_f x$ and by Obs.~\ref{obs:merge-smoothing} it satisfies $f_\e(y)=f(x)-\e$, thus preserving the labeling $\ell$ in particular.
\end{proof}
This poset with a flow is denoted by $(\mathbf{PhTree}_{n},U)$ and the corresponding interleaving distance by $d_{\mathbf{PhTree}_n}$.

\section{The $\ell^\infty$-Cophenetic Metric as an Interleaving Distance}
\label{sec:CopheneteicIsInterleaving}
The information contained in a phylogenetic tree with $n$ leaves can be stored to a vector in $\R^{n(n+1)/2}$
known as the \textbf{cophenetic vector} and provide a collection of  $\ell^p$-type metrics for phylogenetic trees.
This was developed in detail by Cardona et al.~\cite{Cardona2013}. 
In this section we realize the $\ell^\infty$-version of this metric as an interleaving distance.

First we define the $\ell^\infty$-cophenetic metric.
\begin{dfn}[\cite{Cardona2013}]
	Let $(\X,f,\ell)$ be a phylogenetic tree with $n$-leaves. 
	To this tree, we associate the \textbf{cophenetic vector} 
	\begin{equation*}
	\Cfunc(\X,f,\ell):=\bigg(f\big(\ell(i)\vee\ell(j)\big)\bigg)_{1\leq i\leq j\leq n}.
	\end{equation*}
	The map $\Cfunc:\mathbf{PhTree}_n\to\R^{n(n+1)/2}$, $(\X,f,\ell)\mapsto\Cfunc(\X,f,\ell)$ is called the \textbf{cophenetic map}.
	The \textbf{$\ell^\infty$-cophenetic metric} between two trees is
	\begin{equation*}
	d_C((\X,f,\ell), (\Y,g,\mu)) = \|\Cfunc(\X,f,\ell)- \Cfunc (\Y,g,\mu)\|_\infty.
	\end{equation*}
\end{dfn}
See Fig.~\ref{fig:ExampleCophMergeTrees} for an example of this construction for two elements of $\mathbf{PhTree}_4$.
\begin{figure}[tb]
	\centering
	\includegraphics[width = .4\textwidth]{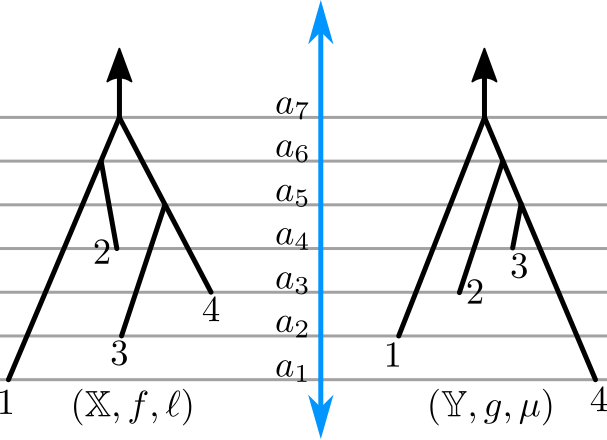}
	\qquad
	\begin{minipage}
	{.35\textwidth}
	\vspace{-1.5in}
	$ \Cfunc(\X,f,\ell) = 
	\begin{pmatrix}
	a_1 & a_6 & a_7 & a_7 \\ 
	\cdot & a_4 & a_7 & a_7\\
	\cdot & \cdot & a_2 & a_5\\
	\cdot & \cdot & \cdot & a_3
	\end{pmatrix}$

	$ \Cfunc(\Y,g,\mu) = 
	\begin{pmatrix}
	a_2 & a_7 & a_7 & a_7 \\ 
	\cdot & a_3 & a_6 & a_6\\
	\cdot & \cdot & a_4 & a_5\\
	\cdot & \cdot & \cdot & a_1
	\end{pmatrix}$
	\end{minipage}

	\caption{An example of two phylogenetic tress in $\mathbf{PhTree}_4$.  The associated cophenetic vectors are shown at right.  If $a_i = i$ for all $i$, then the $\ell^\infty$-cophenetic distance between the trees is 2.}
	\label{fig:ExampleCophMergeTrees}
\end{figure}
Thinking of $\R^{n(n+1)/2}$ as a poset $(\R^{n(n+1)/2},\geq)$, 
we wish to extend $\Cfunc$ to a functor by providing an assignment on morphisms.
Since these categories are both posets, we need only ensure that if $(\X,f,\ell)\preceq(\Y,g,\mu)$, then $\Cfunc(\X,f,\ell)\geq\Cfunc(\Y,g,\mu)$. 
Assume $(\X,f,\ell)\preceq(\Y,g,\mu)$ and thus the unique function preserving and label preserving map $(\X,f,\ell)\xrightarrow{\varphi}(\Y,g,\mu)$. 
By Prop.~\ref{prop:MorphismProperties}, for any $1 \leq i \leq j \leq n$, 
\begin{align*}
f\big(\ell(i)\vee\ell(j)\big)
&=g\circ\varphi\big(\ell(i)\vee\ell(j)\big)\\
&\geq g\big(\mu(i)\vee\mu(j)\big)
\end{align*}
which implies the required inequality in $(\R^{n(n+1)/2},\geq)$.


Consider the strict flow $\Omega$ on the poset $\R^{n(n+1)/2}=(\R^{n(n+1)/2},\geq)$ given by the $\e$-shift downward $\Omega_\e:\R^{n(n+1)/2}\to\R^{n(n+1)/2}$, 
$(r_{i,j})_{1\leq i\leq j\leq n}\mapsto(r_{i,j}-\e)_{1\leq i\leq j\leq n}$.
We can now show that the cophenetic metric of \cite{Cardona2013} is a realization of an interleaving distance.

\begin{thm}
		\label{thm:cophenetic-embedding}
	The cophenetic map $\Cfunc$ forms a fully faithful strict equivariant functor 
	\begin{equation*}
		\Cfunc:(\mathbf{PhTree}_n,U)\to(\R^{n(n+1)/2},\Omega)
	\end{equation*}
 between posets with a flow.
\end{thm}
\begin{proof}
Let $(\X,f,\ell)$ be a phylogenetic tree (and thus a merge tree in particular).
Let $\e\geq0$ and let $U_\e(\X,f,\ell)=(\X_\e,f_\e,\ell_\e)$ be the corresponding phylogenetic tree obtained after the topological $\e$-smoothing.
By Obs.~\ref{obs:merge-smoothing} we get $f_\e(\ell_\e(i)\vee\ell_\e(j))=f(\ell(i)\vee\ell(j))-\e$ for all $1\leq i\leq j \leq n$.
Hence we have:
	\begin{align*}
	\Cfunc U_\e(\X,f,\ell)&=\Cfunc(\X_\e,f_\e,\ell_\e)\\
	&=(f_\e(\ell_\e(i)\vee\ell_\e(j)))_{1\leq i\leq j\leq n}\\
	&=(f(\ell(i)\vee\ell(j))-\e)_{1\leq i\leq j\leq n}\\ 
	&=\Omega_\e\Cfunc(\X,f,\ell).
	\end{align*} 
Thus we obtain the identity map $\Cfunc U_\e(\X,f,\ell)\to\Omega_\e\Cfunc(\X,f,\ell)$. 
Hence by Prop.~\ref{Prop:Poset} the cophenetic map $\Cfunc$ is a fully faithful strict equivariant functor $\Cfunc:\mathbf{PhTree}_n\to(\R^{n(n+1)/2},\geq)$.
\end{proof}

\begin{cor}
	\label{cor:cophenetic-iso}
	The cophenetic map forms an isometric embedding 
	\begin{align*}
		\Cfunc:\bigg(\Ob(\mathbf{PhTree}_n),d_{\mathbf{PhTree}_n}\bigg)
		\to
		\bigg (\Ob(\R^{n(n+1)/2}),||\cdot||_{\infty}\bigg)
	\end{align*}
with respect to the interleaving distances.
	As a result, we obtain the following formula for $d_{\mathbf{PhTree}_n}$:
	$$d_{\mathbf{PhTree}_n}((\X,f,\ell),(\Y,g,\mu))=\max_{1\leq i\leq j\leq n}|f(\ell(i)\vee\ell(j))-g(\mu(i)\vee\mu(j))|$$
	for every $(\X,f,\ell),(\Y,g,\mu)$ in $\mathbf{PhTree}_n$.
\end{cor}
\begin{proof}
	The proof follows directly from Cor.~\ref{cor:Isometric Embedding}.
	The formula follows from combining the definition of $\Cfunc$ with $\|\cdot \|_\infty$.
\end{proof}
This formula shows that the interleaving distance on phylogenetic trees can be computed in $O(n^2)$ time, where $n$ is the number of leaves.


\section{Discussion}
\label{sec:Discussion}

In this paper, we have shown that the $\ell^\infty$-cophenetic metric for phylogenetic trees can be realized as an interleaving distance.  
While from the outside, it might look as if we are taking a massive, formalistic hammer to a simple problem, but there are good reasons for this viewpoint.
Namely, viewing phylogenetic trees as objects of a category with an interleaving distance means that we can extend these ideas to more complicated structures. 
In particular, there is increasing interest in understanding not just phylogenetic trees but phylogenetic networks, and we believe that this interleaving for the tree case can be extended immediately to provide an option for comparison of these structures. 
It is worth observing that the idea of lattices will likely be useful if one wants to construct similar construction on Reeb graphs.
From the side of the study of interleaving distances, this special case gives an example of an interleaving distance which is polynomial time computable.
However, in its most general form, the Reeb graph interleaving distance is graph isomorphism hard.
Thus, we expect there is something to be learned from this special case which can provide either approximation methods or some sort of fixed parameter tractable algorithm to better understand the difficulties inherent in computing this metric.


\bibliographystyle{ieeetr}
\bibliography{Phylo}

\begin{thebibliography}{10}

\bibitem{Diaconis1998}
P.~W. Diaconis and S.~P. Holmes, ``Matchings and phylogenetic trees,'' {\em
  Proceedings of the National Academy of Sciences}, vol.~95, no.~25,
  pp.~14600--14602, 1998.

\bibitem{Robinson1981}
D.~F. Robinson and L.~R. Foulds, ``Comparison of phylogenetic trees,'' {\em
  Mathematical Biosciences}, vol.~53, no.~1, pp.~131--147, 1981.

\bibitem{Robinson1979}
D.~Robinson and L.~Foulds, ``Comparison of weighted labelled trees,'' in {\em
  Combinatorial mathematics VI}, pp.~119--126, Springer, 1979.

\bibitem{Bryant2000}
D.~Bryant, J.~Tsang, P.~E. Kearney, and M.~Li, ``Computing the quartet distance
  between evolutionary trees,'' in {\em Symposium on Discrete Algorithms:
  Proceedings of the eleventh annual ACM-SIAM symposium on Discrete
  algorithms}, vol.~9, pp.~285--286, 2000.

\bibitem{Billera2001}
L.~J. Billera, S.~P. Holmes, and K.~Vogtmann, ``Geometry of the space of
  phylogenetic trees,'' {\em Advances in Applied Mathematics}, vol.~27, no.~4,
  pp.~733--767, 2001.

\bibitem{Owen2011}
M.~Owen, ``Computing geodesic distances in tree space,'' {\em SIAM Journal on
  Discrete Mathematics}, vol.~25, no.~4, pp.~1506--1529, 2011.

\bibitem{Cardona2013}
G.~Cardona, A.~Mir, F.~Rossell{\'o}, L.~Rotger, and D.~S{\'a}nchez,
  ``Cophenetic metrics for phylogenetic trees, after sokal and rohlf,'' {\em
  BMC bioinformatics}, vol.~14, no.~1, p.~3, 2013.

\bibitem{Mailund2004}
T.~Mailund and C.~N. Pedersen, ``Qdist---quartet distance between evolutionary
  trees,'' {\em Bioinformatics}, vol.~20, no.~10, pp.~1636--1637, 2004.

\bibitem{Alberich2009}
R.~Alberich, G.~Cardona, F.~Rossell{\'o}, and G.~Valiente, ``An algebraic
  metric for phylogenetic trees,'' {\em Applied Mathematics Letters}, vol.~22,
  no.~9, pp.~1320--1324, 2009.

\bibitem{Fernau2010}
H.~Fernau, M.~Kaufmann, and M.~Poths, ``Comparing trees via crossing
  minimization,'' {\em Journal of Computer and System Sciences}, vol.~76,
  no.~7, pp.~593--608, 2010.

\bibitem{Moulton2015}
V.~Moulton and T.~Wu, ``A parsimony-based metric for phylogenetic trees,'' {\em
  Advances in Applied Mathematics}, vol.~66, pp.~22--45, 2015.

\bibitem{Valiente2001}
G.~Valiente, ``An efficient bottom-up distance between trees,'' in {\em spire},
  p.~0212, IEEE, 2001.

\bibitem{Reeb1946}
G.~Reeb, ``Sur les points singuliers d'une forme de pfaff compl\`{e}ment
  int\'{e}grable ou d'une fonction num\'{e}rique.,'' {\em Comptes Rendus de
  L'Acad\'{e}mie ses S\'{e}ances}, vol.~222, pp.~847--849, 1946.

\bibitem{Biasotti2008}
S.~Biasotti, D.~Giorgi, M.~Spagnuolo, and B.~Falcidieno, ``Reeb graphs for
  shape analysis and applications,'' {\em Theoretical Computer Science:
  Computational Algebraic Geometry and Applications}, vol.~392, no.~13, pp.~5
  -- 22, 2008.

\bibitem{deSilva2016}
V.~de~Silva, E.~Munch, and A.~Patel, ``Categorified {R}eeb graphs,'' {\em
  Discrete \& Computational Geometry}, pp.~1--53, 2016.

\bibitem{Bauer2015}
U.~Bauer, E.~Munch, and Y.~Wang, ``Strong equivalence of the interleaving and
  functional distortion metrics for {R}eeb graphs,'' in {\em 31st International
  Symposium on Computational Geometry (SoCG 2015)} (L.~Arge and J.~Pach, eds.),
  vol.~34 of {\em Leibniz International Proceedings in Informatics (LIPIcs)},
  (Dagstuhl, Germany), pp.~461--475, Schloss Dagstuhl--Leibniz-Zentrum fuer
  Informatik, 2015.

\bibitem{DiFabio2016}
B.~{Di Fabio} and C.~Landi, ``The edit distance for {R}eeb graphs of
  surfaces,'' {\em Discrete {\&} Computational Geometry}, vol.~55,
  pp.~423--461, jan 2016.

\bibitem{BauerDiFabioLandi2016}
U.~Bauer, B.~Di~Fabio, and C.~Landi, ``An edit distance for {R}eeb graphs,''
  2016.

\bibitem{Beketayev2014}
K.~Beketayev, D.~Yeliussizov, D.~Morozov, G.~H. Weber, and B.~Hamann,
  ``Measuring the distance between merge trees,'' in {\em Mathematics and
  Visualization}, pp.~151--165, Springer International Publishing, 2014.

\bibitem{Morozov2013}
D.~Morozov, K.~Beketayev, and G.~Weber, ``Interleaving distance between merge
  trees,'' in {\em Proceedings of TopoInVis}, 2013.

\bibitem{Agarwal2015a}
P.~K. Agarwal, K.~Fox, A.~Nath, A.~Sidiropoulos, and Y.~Wang, ``Computing the
  {G}romov-{H}ausdorff distance for metric trees,'' {\em arXiv:1509.05751},
  2015.

\bibitem{Bauer2014}
U.~Bauer, X.~Ge, and Y.~Wang, ``Measuring distance between {R}eeb graphs,'' in
  {\em Proceedings of the Thirtieth Annual Symposium on Computational Geometry
  - SoCG '14, Kyoto, Japan}, 2014.

\bibitem{singh2007topological}
G.~Singh, F.~M{\'e}moli, and G.~E. Carlsson, ``Topological methods for the
  analysis of high dimensional data sets and {3D} object recognition.,'' in
  {\em SPBG}, pp.~91--100, 2007.

\bibitem{Munch2016}
E.~Munch and B.~Wang, ``Convergence between categorical representations of
  {R}eeb space and mapper,'' in {\em 32nd International Symposium on
  Computational Geometry (SoCG 2016)} (S.~Fekete and A.~Lubiw, eds.), vol.~51
  of {\em Leibniz International Proceedings in Informatics (LIPIcs)},
  (Dagstuhl, Germany), pp.~53:1--53:16, Schloss Dagstuhl--Leibniz-Zentrum fuer
  Informatik, 2016.

\bibitem{Carriere2017b}
M.~Carri{\`{e}}re and S.~Oudot, ``Structure and stability of the
  one-dimensional mapper,'' {\em Foundations of Computational Mathematics}, oct
  2017.

\bibitem{Babu2013}
A.~Babu, {\em Zigzag Coarsenings, Mapper Stability and Gene network Analyses}.
\newblock PhD thesis, Stanford University, 2013.

\bibitem{Chazal2009b}
F.~Chazal, D.~Cohen-Steiner, M.~Glisse, L.~J. Guibas, and S.~Y. Oudot,
  ``Proximity of persistence modules and their diagrams,'' in {\em Proceedings
  of the 25th annual symposium on Computational geometry}, SCG '09, (New York,
  NY, USA), pp.~237--246, ACM, 2009.

\bibitem{Chazal2016}
F.~Chazal, V.~de~Silva, M.~Glisse, and S.~Oudot, {\em The Structure and
  Stability of Persistence Modules}.
\newblock Springer International Publishing, 2016.

\bibitem{Bubenik2014}
P.~Bubenik and J.~A. Scott, ``Categorification of persistent homology,'' {\em
  Discrete \& Computational Geometry}, vol.~51, no.~3, pp.~600--627, 2014.

\bibitem{Bubenik2014a}
P.~Bubenik, V.~de~Silva, and J.~Scott, ``Metrics for generalized persistence
  modules,'' {\em Foundations of Computational Mathematics}, 2014.

\bibitem{AnastasiosThesis}
A.~Stefanou, ``Dynamics on categories and applications.'' Preprint, 2018.

\bibitem{Silva2017}
V.~de~Silva, E.~Munch, and A.~Stefanou, ``Theory of interleavings on
  $[0,\infty)$-actegories,'' {\em arXiv:1706.04095}, 2017.

\bibitem{Curry2014}
J.~Curry, {\em Sheaves, Cosheaves and Applications}.
\newblock PhD thesis, University of Pennsylvania, 2014.

\bibitem{sidiropoulos9472computing}
A.~Sidiropoulos and Y.~Wang, ``Computing the {G}romov-{H}ausdorff distance for
  metric trees,'' {\em Algorithms and Computation LNCS 9472}, p.~529.

\bibitem{bjerkevik2017computational}
H.~B. Bjerkevik and M.~B. Botnan, ``Computational complexity of the
  interleaving distance,'' {\em arXiv:1712.04281}, 2017.

\bibitem{Eldridge2015a}
J.~Eldridge, M.~Belkin, and Y.~Wang, ``Beyond {H}artigan consistency: {M}erge
  distortion metric for hierarchical clustering.,'' in {\em COLT},
  pp.~588--606, 2015.

\bibitem{MacLane1978}
S.~Mac~Lane, {\em Categories for the Working Mathematician}.
\newblock Springer-Verlag New York, 2~ed., 1978.

\end{thebibliography}
	
\end{document}